\documentclass[11pt]{article}
\pdfoutput=1
\usepackage[%
compact=,%
paper=a4paper,%
largepaper=true,%
marginfrac=12,%
runningtitle=false,%
halfparskip,%
font=LatinModern,%
biblatex=false,%
natbib=true,%
bibliographystyle=plainurl,%
theorems=numbersfirst,%
theoremswithin=document,%
]{ifiseries}
\usepackage[graphtheory,gametheory,complexitytheory]{newKit}
\usepackage{enumitem}
\newcommand*{\nn}{\nu}
\newcommand*{\sepmax}{\sep_{\max}}
\newcommand*{\Emax}{E_{\max}}
\newcommand*{\hG}{\widehat{G}}
\renewcommand*{\cost}{\cC}
\newcommand*{\Vmax}{V_{\max}}
\newcommand*{\nmax}{n_{\max}}

\begin{document}

\title{Swap Equilibria under Link and Vertex Destruction}
\author{Lasse Kliemann \and Elmira Shirazi Sheykhdarabadi \and Anand Srivastav}
\date{\smaller%
  Department of Computer Science\\Kiel University\\Christian\-/Albrechts\-/Platz 4\\24118 Kiel, Germany\\
  \texttt{\{lki,esh,asr\}@informatik.uni-kiel.de}}
\maketitle

\begin{abstract}
  \noindent
  We initiate the study of the \term{destruction model} (\aka \term{adversary model})
  introduced by Kliemann (2010), using the stability concept of \term{swap equilibrium}
  introduced by Alon \etal~(2010).
  The destruction model is a network formation game incorporating the robustness of a network
  under a more or less targeted attack.
  In addition to bringing in the swap equilibrium (SE) concept,
  we extend the model from an attack on the edges of the network
  to an attack on its vertices.
  Vertex destruction can generally cause more harm and tends to be more difficult to analyze.
  \par\smallskip\noindent
  We prove structural results and linear upper bounds or super\-/linear lower bounds on the social cost
  of SE under different attack scenarios.
  The most complex case is when the vertex to be destroyed is chosen uniformly at random
  from the set of those vertices where each causes a maximum number of player pairs to be separated
  (called a max\-/sep vertex).
  We prove a lower bound on the social cost of $\Ohm{n^{3/2}}$ for this case
  and initiate an understanding of the structural properties of SE in this scenario.
  Namely, we prove that there is no SE that is a tree and has only one max\-/sep vertex.
  We conjecture that this result can be generalized,
  in particular we conjecture that there is no SE that is a tree.
  On the other hand, we prove that if the vertex to be destroyed is chosen uniformly at random
  from the set of \emphasis{all} vertices,
  then each SE is a tree (unless it is two\-/connected).
  Our conjecture would imply that moving from the uniform probability measure
  to a measure concentrated on the max\-/sep vertices,
  means moving from no SE having a cycle (unless two\-/connected)
  to each SE having a cycle.
  This would ask for a more detailed study of this transition in future work.
\end{abstract}

\section{Introduction}

Game theoretic models for the study of the decentralized formation of networks gained remarkable attention during the past two decades.
In most of those models, the vertices of a graph correspond to players,
and each player's choice of action can have an influence on the structure of the graph.
Models differ regarding the different actions available to the players
and regarding the criteria under which the quality of the graphs is evaluated.
For the latter, most models have focused on centrality\-/type criteria,
for example, players aim to minimize the sum of distances over all other players.
Recently, robustness aspects have been addressed in the form of the
\term{destruction model} (\aka \term{adversary model})~\cite{Kli10b,Kli11a,Kli14,Kli16}.
In this model, players anticipate the destruction of exactly one edge in the graph,
and the cost function for each player~$v$ gives the expected number of other players that~$v$ will no longer be able to reach
after the destruction.
Social cost is the sum over all players' costs, which is equal to the expected number of separated vertex pairs,
that is, the expected number of all ordered pairs $(v,w)$
such that there is no path anymore between $v$ and $w$ after destruction took place.
The model allows many variations since the edge to destroy is determined randomly
according to a probability measure that may even depend on the graph (this dependence is known to the players).\footnote{%
  The name \enquote{adversary model} is the original one.
  However, \enquote{adversary} was found to be more suited
  to describe an entity that aims at maximizing cost under equilibrium.
  This is not necessarily the case in our model.}

In this work, we use the stability concept of \term{swap equilibrium} (SE)~\cite{ADHL10} for the destruction model.
Moreover, we extend the model to the destruction of exactly one \emphasis{vertex} instead of an edge.
We will henceforth speak of \term{edge destruction model} and \term{vertex destruction model}
in order to distinguish the two.

\subsection*{Previous and Related Work}

Network formation games date back to the 1990ies, see Jackson and Wolinsky~\cite{JW96} for an early publication.
In Computer Science, network formation games have gained attention since the work by Fabrikant \etal~\cite{FLM+03} in 2003.
The destruction model was introduced by Kliemann in 2010~\cite{Kli10b}
and subsequently studied in a series of publications~\cite{Kli11a,Kli14,Kli16}.
The focus was on the price of anarchy for Nash equilibrium (NE) and pairwise stability (PS).
Earlier works on robustness in network formation include~\cite{JW96,CFSK04,BG00a,HS03,HS05}.
None of those earlier models allows a structure\-/dependent destruction probability as in the destruction model;
for a detailed discussion we refer to~\cite[Sec.~4]{Kli11a}.

The stability concepts of NE and PS require an edge cost parameter $\al \in \RRpos$.
Computationally deciding whether a graph constitutes an NE is hindered by an exponential search space,
and can indeed be \NPhard~\cite{FLM+03}.
This has raised concerns about the applicability of the model,
since we should not expect players to solve an \NPhard problem.
Therefore, many variations have been introduced in order to make the model more tractable.
Usually, the idea is to limit the choices of the players,
for example, to single\-/edge deviations~\cite{JW96,Len12}.
An even more drastic approach was suggested by Alon \etal in 2010~\cite{ADHL10}.
They removed the $\al$ parameter and considered only edge swaps:
as graph is a \term{swap equilibrium} if no player can improve her cost by removing an incident edge
and then creating a new incident edge -- this action is called an \term{edge swap}.
It is also allowed to simply remove an incident edge without creating a new one,
which can make sense if edges can be harmful, \eg if cost incorporates the risk of contagion.
A variation, introduced by Mihal{\'a}k and Schlegel in 2012~\cite{MS12}, are \term{edge ownerships}:
each edge is owned by exactly one of its endpoints
and may only be swapped (or removed) by its owner.
The resulting stability concept is called \term{asymmetric swap equilibrium} (ASE).
We do not consider edge ownerships or ASE in this work, but this extension is most certainly interesting for future work.

The following three recent publications address robustness in a network formation framework similar to ours:
\begin{compactitemize}
\item Meirom \etal~\cite{MMA15} consider a cost function that uses a linear combination of
  the lengths of two short disjoint paths.
  The idea is that players build a graph where for each shortest path, there is a backup path of reasonable length.
\item Goyal \etal~\cite{GJK16} consider a model where each player,
  in addition to building edges, can choose to immunize herself in exchange for a fee.
  Then an adversary selects a connected component of non-immunized vertices to destroy;
  an alternative description is that the adversary picks a vertex and then the destruction spreads from there
  while immunized vertices act as firewalls.
  A player's utility is the expected size of her connected component after the destruction took place,
  which is~$0$ if the player itself is destroyed.
  This utility is almost exactly the positive version of our cost: if $\cost(v)$ is the expected number of cut-off vertices,
  then $n-\cost(v)$ is the expected size of $v$'s component after the attack.
  (We differ in that if the player itself is attacked, we say that it is cut-off from $n-1$ vertices,
  so its components size is~$1$, not~$0$.)
  However, the kind of destruction is very different from ours.
  Although it can be formulated as an attack on a vertex, the contagious properties of the attack
  move the focus to different connectivity properties of the attacked vertex in comparison to our model.
  For example, if a leaf (that is, a vertex of degree~$1$) is attacked in our model,
  the overall damage is relatively small, namely we have $2(n-1)$ separated vertex pairs.
  In their model however, if the neighbor of the leaf is not immunized, the destruction will spread and
  the overall damage can be much higher.
  \par
  For earlier work on network formation with contagious risk, see~\cite{BEK13}.
\item Chauhan \etal~\cite{CLMM16} extended the edge destruction model by incorporating distances:
  cost for player $v$ is the expected sum of distances to all other players after edge destruction.
\end{compactitemize}

\subsection*{Our Contribution}

We prove quantitative and structural results for two types of destroyers under the stability concept of swap equilibrium (SE):
the \term{uniform destroyer} picks an edge or vertex uniformly at random,
while the \term{extreme destroyer} picks an edge or vertex uniformly at random
from the set of edges or vertices, respectively, where the destruction of each causes a maximum number of vertex pairs to be separated.
An edge or vertex that does the latter is called a \term{max-sep edge} or \term{max-sep vertex}, respectively.
For edge destruction, we also consider some variations, most notably a destroyer that chooses an edge
from the set of bridges uniformly at random.
For vertex destruction, we also consider the case that the probability for destruction of a vertex $v$ is proportional to its degree $\deg(v)$,
which we call the \term{degree\-/proportional destroyer}.

We prove that for uniform and extreme edge destruction and uniform bridge destruction,
an SE is bridgeless or has a star\-/like structure.
A consequence of this is that in terms of social cost, those SE are very efficient,
namely social cost of any of those SE is $\Oh{n}$.
For uniform vertex destruction, we prove that if an SE is not two\-/connected, then it is a tree.
This again implies an $\Oh{n}$ bound on the social cost.

For the degree\-/proportional vertex destroyer, the situation is very different.
Social cost of an SE can be as high as $\Ohm{n^2}$, which is the highest order possible in this model.
This lower bound is attained on a simple graph, namely a star.

For extreme vertex destruction, we also give a super-linear lower bound on the social cost of SE,
namely $\Ohm{n^{3/2}}$.
The construction is still roughly star\-/like, but more complicated:
we need a clique of certain size at the center to which paths up to a certain length are attached.

Finally, we prove a structural result for extreme vertex destruction:
if $n \geq 8$, there is no tree SE with only one max\-/sep vertex.
This means that in a tree SE (with $n \geq 8$), the destroyer will always have at least two vertices to choose from.

\subsection*{Ongoing Work and Open Problems}

We have extensive experimental evidence and also indications on the theory side
that our structural result for the extreme vertex destroyer
(that there is no SE tree with one max\-/sep vertex for $n \geq 8$) can be extended in two ways.
We conjecture for $n=\Ohm{1}$:
\begin{enumerate}[label=(\roman*)]
\item There is no SE graph under extreme vertex destruction with only one max\-/sep vertex.
  (This extends our result for trees to general graphs.)
\item There is no SE graph under extreme vertex destruction that is a tree.
  (This extends our non\-/existence result for trees with one max\-/sep vertex
  to trees in general.)
\end{enumerate}
We expect to prove at least one of those two conjectures in the near future.

Recall that we prove
that unless the graph is two\-/connected, an SE cannot contain cycles for the uniform vertex destroyer.
On the other hand, for the extreme vertex destroyer, our conjecture (ii) would imply that an SE is only possible if we have at least one cycle.
It would then be interesting to find a family $\parens{\cD^{\veps}}_{\veps \in [0,1]}$ of vertex destroyers
such that $\cD^0$ is the uniform destroyer and $\cD^1$ is the extreme destroyer,
and the others are something suitable in between.
When moving $\veps$ from $0$ to~$1$, the probability measure should concentrate more and more
on the max-sep vertices.
It would then be interesting to find out at which values of $\veps$ the situation switches from
no non\-/two\-/connected SE having a cycle to each SE having a cycle.
Our conjecture (ii) would imply that it must switch an odd number of times.

Another interesting direction is the extension to edge ownerships and the asymmetric swap equilibrium.
Closing the gap between our lower $\Ohm{n^{3/2}}$ bound and the trivial $\Oh{n^2}$ upper bound
for social cost of extreme vertex destruction is also an open task.

\section{Model and Notation}

Fix the number of players $n \in \NN_{\geq 3}$.
All our graphs are finite, simple, and undirected.
The undirected edge $\set{v,w}$ between vertices $v$ and $w$ is denoted $vw$ or $wv$.
Denote $\cG_n$ the set of all graphs on the vertex set $V_n \df \setn{n} \df \setft{1}{n}$.
We use the term \term{player} and \term{vertex} synonymously for graphs in $\cG_n$.
A~\term{swap} for a graph $G \in \cG_n$ is a triple of players $(a,b,c)$ such that
$ab \in E(G)$ and $ac \not \in E(G)$.
Denote $S(G) \subseteq V_n^3$ the set of all swaps of~$G$.
Denote $G+(a,b,c)$ the graph that is obtained from $G$ by removing $ab$ and inserting $ac$;
we say that player $a$ swaps her edge $ab$ for the new edge $ac$.
A~\term{cost function} $\cost_G$ for $G \in \cG_n$ assigns to each player $v \in V_n$ a number $\cost_G(v) \in \RRnn$.
The \term{social cost} of $G$ is $\SC(G) \df \sum_{v \in V_n} \cost_G(v)$.
We call $G$ a \term{swap equilibrium} (SE)
if $\cost_G(a) \leq \cost_{G+(a,b,c)}(a)$ for all $(a,b,c) \in S(G)$
and $\cost_G(a) \leq \cost_{G-ab}(a)$ for all $ab \in E(G)$.

\textbf{Edge Destruction.}
Let $G \in \cG_n$ be connected.
For $v, w \in V_n$ denote $\cR_G(v,w)$ the set of all $v$-$w$ paths in $G$.
The \term{relevance} of $e \in E(G)$ for $v \in V_n$ is
\begin{equation*}
  \rel_G(e,v) \df \card{\set{w \in V_n \suchthat \forall P \in \cR(v,w) \holds e \in E(P)}} \comma
\end{equation*}
that is, the number of vertices for which in order to reach them from $v$, we necessarily have to traverse edge~$e$.
When $e$ is removed from the graph, then there will be exactly $\rel_G(e,v)$ vertices
that $v$ will no longer be able to reach;
we also say that those vertices are \term{cut-off} from~$v$ or that $v$ is \term{cut-off} from them.
An~\term{edge destroyer} $\cD$ is a map that associates with each $G \in \cG_n$ a probability measure $\cD_G$ on $E(G)$,
that is, $\cD_G(e) \in [0,1]$ for each $e$, and \mbox{$\sum_{e \in E(G)} \cD_G(e) = 1$}.
Given $\cD$, we define the \term{cost} for player $v$ in $G$ as
\begin{equation*}
  \cost_G(v) \df \sum_{e \in E} \rel_G(e,v) \, \cD_G(e) \comma
\end{equation*}
that is, the expected number of vertices that $v$ will be cut-off from after one edge is removed
randomly according to the measure $\cD_G$.
If $G$ is disconnected, cost is defined to be~$\infty$.

The \term{separation} $\sep(e)$ of an edge $e \in E(G)$ is the number of ordered player pairs $(v,w)$
such that the removal of $e$ will destroy all $v$-$w$ paths in~$G$.
If $e$ is a non-bridge, then clearly $\sep(e) = 0$.
Otherwise, $G-e$ has exactly two components, say $K_1,K_2 \subseteq V_n$,
and we call $\nn(v) \df \min\set{\card{K_1}, \card{K_2}}$ the \term{minimal\-/component size} of $e$.
Then $\sep(e) = 2 \nn(e) (n-\nn(e))$.
It is easy to see that $\SC(G) = \sum_{e \in E(G)} \sep(e) \, \cD_G(e)$.

\textbf{Vertex Destruction.}
For the vertex destruction model, the destroyer $\cD$ associates each $G \in \cG_n$
with a probability measure on the vertices of $G$, that is, $\cD_G(v) \in [0,1]$ for each $v \in V_n$, and \mbox{$\sum_{v \in V_n} \cD_G(v) = 1$}.
We define the destruction of a vertex not as its removal from the graph,
but as the removal of all its incident edges.
This is reflected by the following definition of relevance and cost.
The \term{relevance} of $u \in V_n$ for $v \in V_n$ is
\begin{equation*}
  \rel_G(u,v) \df \card{\set{w \in V_n \suchthat \forall P \in \cR(v,w) \holds u \in V(P)}} \period
\end{equation*}
Note that since $v$ is in every $v$-$w$ path, we have $\rel_G(v,v) = n-1$,
which is exactly the number of vertices that will be cut-off from $v$ if all edges incident with $v$ are removed.
Given a vertex destroyer $\cD$, we define the cost for player $v$ in $G$ as
\begin{equation*}
  \cost_G(v) \df \sum_{u \in V} \rel_G(u,v) \, \cD_G(u) \period
\end{equation*}
Again, a disconnected graph induces infinite cost for each player.

The \term{separation} $\sep(u)$ of a vertex $u \in V_n$ is the number of ordered player pairs $(v,w)$
such that the removal of $u$ will destroy all $v$-$w$ paths in~$G$.
It is again easy to see that $\SC(G) = \sum_{u \in V_n} \sep(u) \, \cD_G(u)$.
Moreover, if removal of $u$ creates $k$ components (not counting $u$ itself) of sizes $\eli{\bet}{k}$,
we have
\begin{equation}
  \label{eq:sep-beta}
  \sep(u) = n^2 - 1 - \sum_{i=1}^k \bet_i^2 \period
\end{equation}

When the graph $G$ is clear from context, we omit the $G$ subscripts or arguments.
When a graph $G'$ is defined, we write $\cost'$, $\rel'$, $\SC'$ \etc instead of $\cost_{G'}$, $\rel_{G'}$, $\SC(G')$, \etc, respectively.
The same goes for~$G''$.

For any connected graph $G=(V,E)$, we call $I \subseteq V$ an \term{island} if it is inclusion-maximal under the condition that
the induced subgraph $G[I]$ is bridge-free (it does not matter whether we mean bridges of $G$ or bridges of $G[I]$).
The \term{bridge tree} $\tiG$ of $G$ is obtained by collapsing each island $I$ of $G$ to a single vertex $\tiI$ and
inserting an edge between $\tiI$ and $\tiJ$ if and only if an edge runs in $G$ between a vertex of $I$ and a vertex of $J$.
Obviously, $IJ \mapsto \tiI \tiJ$ is a bijection between the set of bridges of $G$ and the set of bridges of $\tiG$,
and we will often identify those two sets.
We refer to \cite{Kli11a} for a more formal treatment of the bridge tree.

We will also use the better-known \term{block-cutvertex} tree, see, \eg, the book by Diestel~\cite{Die05} for a definition.

\section{Uniform Edge and Uniform Bridge Destruction}

Denote $B(G) \subseteq E(G)$ the set of all the bridges of $G \in \cG_n$
and
\begin{equation*}
  S^B(G) \df \set{(a,b,c) \in S(G) \suchthat ab \in B(G) \land \text{$G+s$ is connected}}
\end{equation*}
the set of \term{bridge swaps} (the case that $G+s$ is disconnected is not interesting since such a swap can never bring an improvement).
Clearly, $ac \in B(G+s)$ for each $(a,b,c) \in S^B(G)$.

An edge destroyer $\cD$ is called the \term{uniform edge destroyer}
if $\cD_G(e) = \frac{1}{\card{E(G)}}$ for each $e \in E(G)$ and each $G \in \cG_n$.
It is called the \term{uniform bridge destroyer}
if $\cD_G(e) = \frac{1}{\card{B(G)}}$ for each $e \in B(G)$ 
and $\cD_G(e) = 0$ for each $e \in E(G) \setminus B(G)$,
for each $G \in \cG_n$.
(If $B(G) = \emptyset$, the graph is $2$-edge-connected and we can take any probability measure for $\cD_G$
since cost for each player is~$0$ in any case.)

In order to point out what are some of the essential properties of those destroyers,
we look at more general destroyers first.
Consider the following condition on a destroyer:
\begin{equation}
  \label{eq:sep}
  \begin{split}
    &\forall G \in \cG_n \innerholds
    \forall s=(a,b,c) \in S^B(G) \holds\\
    &\phantom{\land}\quad\sep_G(ab) = \sep_{G+s}(ac) \implies \cD_G(ab) = \cD_{G+s}(ac)\\
    &\land\quad\forall e \in E(G) \cap E(G+s) \holds
    \sep_G(e) = \sep_{G+s}(e) \implies \cD_G(e) = \cD_{G+s}(e)
  \end{split}
\end{equation}
This means that if after a bridge swap an edge maintains its separation, then it also maintains its probability.
This clearly includes the uniform edge destroyer and the uniformly bridge destroyer.
The next proposition shows that \eqref{eq:sep} is equivalent to the following simpler condition:
\begin{equation}
  \label{eq:const}
  \begin{split}
    &\forall G \in \cG_n \innerholds
    \forall s=(a,b,c) \in S^B(G) \holds\\
    &\phantom{\land}\quad\cD_G(ab) = \cD_{G+s}(ac)\\
    &\land\quad\forall e \in E(G) \cap E(G+s) \holds
    \cD_G(e) = \cD_{G+s}(e)
  \end{split}
\end{equation}
This means that each edge carries its fixed probability that -- in case of a bridge -- sticks to it even when it is swapped for another edge.

\begin{proposition}
  (\ref{eq:sep}) and (\ref{eq:const}) are equivalent.
\end{proposition}
\begin{proof}
  It is clear that (\ref{eq:const}) implies (\ref{eq:sep}).
  So let $\cD$ be a destroyer with property (\ref{eq:sep}).
  Let $G \in \cG_n$ and $s=(a,b,c) \in S^B(G)$.
  Since $\nn_G(ab) = \nn_{G+s}(ac)$,
  we have $\sep_G(ab) = \sep_{G+s}(ac)$, hence $\cD_G(ab) = \cD_{G+s}(ac)$.
  Now, consider the special case first that $(a,b,c)$ is a path in the bridge tree.
  Then the only separation that changes due to $s$ is that of $bc$.
  Since separations of all other edges are maintained, they also maintain their probabilities.
  Since all the probabilities add up to $1$, the edge $bc$ also maintains its probability.
  In the general case, we have a path $(v_0=b,v_1,\hdots,v_k=c)$.
  Conducting the sequence of swaps $(a,v_0,v_1), (a,v_1,v_2), \hdots, (a,v_{k-1},v_k)$
  gives the graph $G+s$, and in each step the edges maintain their probability.
\end{proof}

The following proof uses the basic idea from~\cite[Thm.~1]{ADHL10}.
\begin{lemma}
  \label{lem:linkdiam2}
  Let $\cD$ be a destroyer with property (\ref{eq:const}).
  Then the bridge tree of an SE with respect to $\cD$ has diameter at most~$2$.
\end{lemma}
\begin{proof}
  Let $G$ be an SE and
  for contradiction assume that $(a,b,c,d)$ is a path in its bridge tree.
  Denote $n_a, n_b, n_c, n_d$ the numbers of vertices in the subtrees rooted at $a,b,c$ and $d$, respectively,
  hence $n = n_a + n_b + n_c + n_d$.
  Consider the swap $s=(a,b,c)$.
  By (\ref{eq:sep}), we have $\cD_G(ab) = \cD_{G+s}(ac)$, and also all other edges maintain their probabilities.
  We have $\rel_G(ab, a) = \rel_{G+s}(ac, a)$ and for all the other edges,
  from the view of $a$, the only relevance that changes is that of $bc$,
  namely from $n_c + n_d$ to $n_b$.
  Since $G$ is an SE, this means $n_c + n_d \leq n_b$.
  Likewise we consider the swap $(d,c,b)$ and obtain $n_a + n_b \leq n_c$.
  Together this implies $n_a \leq 0$, which is impossible.
\end{proof}

\begin{theorem}
  \label{prop:uniform-edge-star}
  Let $G$ be an SE for the uniform edge destroyer.
  Then $G$ is bridgeless or a star, hence $\SC(G) \leq 2(n-1) = \Oh{n}$.
\end{theorem}
\begin{proof}
  If $G$ is bridgeless, then $\SC(G) = 0$.
  So assume that $G$ contains a bridge.
  We want to show that $G$ is a tree, so for contradiction assume that $G$ is not a tree.
  Let $I$ be an island containing a cycle and let $bc$ be a bridge with $b \in I$.
  Then there is a cycle $C$ in $I$ that traverses~$b$.
  Choose $a$ so that $ab \in E(C)$.
  Then the swap $(a,b,c)$ puts the bridge $bc$ on a cycle and makes it part of the island,
  so its relevance for $a$ drops from a positive value to $0$.
  No new bridges are introduced, and the relevances of all other edges remain the same for player $a$.
  Hence this is an improving swap, a contradiction to SE.

  Since we know that $G$ is a tree, $G$ coincides with its bridge tree.
  By \autoref{lem:linkdiam2}, $\diam(G) \leq 2$.
  Since $n \geq 3$, we conclude that $G$ is a star.
  It follows $\SC(G) = 2(n-1)$.
\end{proof}

\begin{theorem}
  Let $G$ be an SE for the uniform bridge destroyer.
  Then $G$ is bridgeless or $\tiG$ is a star where each of the outer islands has exactly one vertex.\footnote{%
    If the star has only one edge and thus there are exactly two islands,
    this statement means that one of the two islands has exactly one vertex.}
  Hence $\SC(G) \leq 2(n-1) = \Oh{n}$.
\end{theorem}
\begin{proof}
  If $G$ is bridgeless, then $\SC(G) = 0$.
  So assume that $G$ contains a bridge.
  By \autoref{lem:linkdiam2}, $\tiG$ is a star.
  Let $I$ be an island that is not the center of the star (\ie, it is an outer island)
  and that contains more than one vertex.
  Then $I$ contains a cycle.
  By a swap as in the proof of \autoref{prop:uniform-edge-star},
  the one bridge $e$ between the center of the star and $I$ can be put on a cycle.
  For the players in $I$, this is a strict improvement since for them, $e$ had the strictly highest relevance of all bridges in~$G$.
  The statement on the social cost follows since only one vertex can be separated from the rest of the graph
  by the removal of a bridge.
\end{proof}

\section{Extreme Edge Destruction}

For $G \in \cG_n$ denote $\sepmax(G) \df \max_{e \in E(G)} \sep(e)$ and
\begin{equation*}
  \Emax(G) \df \set{e \in E(G) \suchthat \sep(e) = \sepmax(G)} \period
\end{equation*}
We call the edges in $\Emax(G)$ the \term{max-sep edges}.
Recall $\sep(e) = 2 \nn(e) (n-\nn(e))$ and note that $x \mapsto x (n-x)$ is strictly increasing on $[0,n/2]$,
hence $\nn(e) = \nn(e')$ for all $e,e' \in \Emax(G)$.
Moreover if $\nn(e) = \nn(e')$ for some $e \in \Emax(G)$ and $e' \in E(G)$, then $e' \in \Emax(G)$.
In other words: exactly all the edges with maximum $\nn(e)$ are max-sep.

An edge destroyer $\cD$ is called the \term{extreme edge destroyer} if $\cD_G(e) = \frac{1}{\card{\Emax(G)}}$
for each $e \in \Emax(G)$ and $\cD_G(e) = 0$ for each $e \in E(G) \setminus \Emax(G)$.

\begin{theorem}
  Let $G$ be an SE under the extreme edge destroyer.
  Then $G$ is bridgeless or $\tiG$ is a star where each of the outer islands has exactly one vertex.
  Hence $\SC(G) \leq 2(n-1) = \Oh{n}$.
\end{theorem}
\begin{proof}
  The expression for the social cost follows from the structural statement.
  Assume for contradiction that $G$ contains bridges and is not of the stated form.
  
  \textbf{Case 1:} $\Emax = \set{\eli{e}{k}}$ with $k \geq 2$.
  By \cite[Prop.~9.1]{Kli11a}, the max-sep edges form a star in the bridge tree.
  For each $i \in \setn{k}$ denote $K_i \subseteq V_n$ the (unique) minimal component of $G-e_i$,
  and denote $K_0$ the island at the center of the star formed by $\Emax$.
  Then $V_n = \dotbigcup_{i=0}^k K_i$ and $\card{K_i} = \card{K_j}$ for all $i,j \in \setn{k}$.
  By assumption, $\card{K_i} \geq 2$ for each~$i$.
  
  \textbf{Case 1.1:} There is a leaf $a \in K_1$.
  Denote $b$ the neighbor of $a$ and let $v \in K_0$.
  Define $G' \df G + (a,b,v)$.
  When moving from $G$ to $G'$, all the edges $\elix{e}{2}{k}$ and the edges in $G[K_2],\hdots,G[K_k]$
  have their separation maintained.
  Edges in $G[K_1]$ have their separation maintained or reduced since their minimal\-/component size reduces.
  The edge $e_1$ has its separation reduced since its minimal\-/component size reduces.
  The new edge $av$ cannot become max\-/sep since $\nn(av) = 1$ while $\nn(e_2) \geq 2$.
  It follows that $\Emax' = \set{\elix{e}{2}{k}}$.
  We have $\rel'(e_i, a) = \rel(e_i, a) < \rel(e_1, a)$ for all $i \geq 2$.
  Hence $\cost'(a) < \cost(a)$, a contradiction to SE.

  \textbf{Case 1.2:} There is a cycle $C$ in $K_1$.
  Let $ab \in E(C)$ and $v \in K_0$.
  By essentially the same arguments as in Case~1.1, we show that player $a$ improves by the swap $(a,b,v)$.

  \textbf{Case 2:} $\Emax = \set{e_1}$.
  This case is more difficult since after reducing the separation of $e_1$,
  we have no other max\-/sep edges that could act as a reference.
  Denote $K_1$ a component of $G-e_1$ with minimum size (if both components of $G-e_1$ have the same size, then pick one arbitrarily)
  and denote $K_0$ the island containing the endpoint of $e_1$ that is not in~$K_1$.
  
  \textbf{Case 2.1:} There is a leaf $a \in K_1$.
  Denote $b$ the neighbor of $a$ and let $v \in K_0$.
  Define $G' \df G + (a,b,v)$.
  We have $\nn'(e_1) = \nn(e_1) - 1$, so $\sep'(e_1)$ has the next lower possible value below $\sep(e_1)$.
  Hence $\Emax' = \set{e_1, e_2, \hdots, e_k}$ for zero or more additional edges $\elix{e}{2}{k}$.
  Since they all have the same minimum\-/component size, they cannot be in $G[K_1]$,
  hence they form a star with $K_0$ as the center (this is the only way that they can form a star in the bridge tree).
  If $av \not\in \Emax'$, then $\cost'(a) = \nn(e_1) - 1 < n - \nn(e_1) = \cost(a)$, hence the swap is an improvement.
  If $av \in \Emax'$, then $\nn(e_1) - 1 = \nn'(av) = 1$, so $\card{K_1} = 2$ and $n \geq 4$.
  Moreover, $k \geq 2$.
  It follows $\cost'(a) = \frac{1}{k} (n-1 + (k-1) (\nn(e_1) - 1)) = \frac{n-2}{k} + 1$.
  On the other hand $\cost(a) = n-2$.
  If $n \geq 5$ or $k \geq 3$, this implies an improvement.
  The remaining case of $n = 4$ and $k=2$ is impossible since for such $n$, the graph $G'$ is a star and thus $k=3$.

  \textbf{Case 2.2:} There is a cycle $C$ in $K_1$.
  We consider a swap like the one in Case 1.2.
  It is not difficult to see that $\rel'(e,a) < \rel(e_1,a)$ for each $e \in \Emax'$,
  hence the swap is an improvement.
\end{proof}

\section{Uniform and Degree-Proportional Vertex Destruction}

A vertex destroyer $\cD$ is called the \term{uniform vertex destroyer}
if $\cD_G(u) = \frac{1}{n}$ for each $u \in V_n$ and each $G \in \cG_n$.

For any connected graph $G=(V,E)$, denote $\cB(G)$ the set of its blocks
and $A(G) \subseteq V$ the set of its cutvertices (also known as \term{articulation points}).
Denote $\hG$ the block-cutvertex tree of $G$,
that is, $V(\hG) = \cB(G) \dotcup A(G)$ and each edge in $\hG$
runs between a block and a cutvertex, namely $Bv \in E(\hG)$
if $B \in \cB(G)$ and $v \in B \cap A(G)$.
We have $\card{B} \geq 2$ for each $B \in \cB(G)$.
If $\card{B} \geq 3$, then $G[B]$ is two\-/connected;
we also say that $B$ is two\-/connected in~$G$.
Recall also that in a two\-/connected graph, for each vertex $v$ we can find a cycle that visits $v$.

The following remark is proved by standard arguments, which are included here for completeness.

\begin{remark}
  \label{rem:two-paths}
  Let $G=(V,E)$ be any two\-/connected graph.
  \begin{enumerate}[label=(\roman*)]
  \item\label{rem:two-paths:i}
    Let $x,y,v \in V$ be three distinct vertices.
    Then there exist paths $P=(v,\hdots,x)$ and $Q=(v,\hdots,y)$ with $V(P) \cap V(Q) = \set{v}$.
  \item\label{rem:two-paths:ii}
    Let $W \subseteq V$ with $\card{W} \geq 2$ and $v \in V \setminus W$.
    Then there are $x,y \in W$ with $x \neq y$
    and paths $P = (v,\hdots,x)$ and $Q = (v,\hdots,y)$
    with $V(P) \cap V(Q) = \set{v}$ and $V(P) \cap W = \set{x}$ and $V(Q) \cap W = \set{y}$.
  \end{enumerate}
\end{remark}
\begin{proof}
  \ref{rem:two-paths:i}
  Add a new vertex $z$ to the graph and connect it with $x$ and with $y$.
  The resulting graph is again two\-/connected.
  Using the global version of Menger's theorem (see, \eg, \cite[Thm.~3.3.6]{Die05})
  we find two independent (that is, internally vertex\-/disjoint) $v$-$z$ paths.
  Taking subpaths yields the result.

  \ref{rem:two-paths:ii}
  Let $x',y' \in W$, $x' \neq y'$ be any two distinct vertices in $W$.
  By \ref{rem:two-paths:i}, we find $P'=(v,\hdots,x')$ and $Q'=(v,\hdots,y')$ with
  \begin{equation}
    \label{eq:paths-disjoint}
    V(P') \cap V(Q') = \set{v} \period
  \end{equation}
  Let $x$ be the first vertex on $P'$ that is also in $W$.
  Define $P \df (v,\hdots,x)$ as a subpath of $P'$.
  Likewise, let $y$ be the first vertex on $Q'$ that is also in $W$ and define $Q \df (v,\hdots,y)$ as a subpath of $Q'$.
  By \eqref{eq:paths-disjoint} we get $x \neq y$, and the other properties follow from the choice of $x$ and $y$.
\end{proof}

\begin{definition}
  \label{def:extension}
  Let $G=(V,E)$ be any graph and $(B_1,b_1,\hdots,B_k,b_k,B_{k+1})$, $k \geq 1$, be a path in its block\-/cutvertex tree $\hG$.
  Assume $\card{B_1} \geq 3$ and let $C=(b_1,a,\hdots,b_1)$ be a cycle in $B_1$.
  Let $c \in B_{k+1} \setminus \set{b_k}$.
  Then we call the swap $(a,b_1,c)$ a \term{cycle extension} with respect to $(B_1,B_{k+1})$;
  note that $\elix{B}{2}{k}$ and $\eli{b}{k}$ are uniquely determined by the pair $(B_1,B_{k+1})$ since $\hG$ is a tree.
\end{definition}

The name \enquote{cycle extension} is chosen since the cycle $C$ is extended into a larger cycle,
thereby merging the blocks that are traversed by the new cycle.
The merging property is proved in the next proposition.

\begin{proposition}
  \label{prop:merge}
  With notation as in \autoref{def:extension}, denote $G' \df G + (a,b_1,c)$.
  Then in $G'$, all the blocks $\eli{B}{k+1}$ are merged into one block and the remaining blocks are maintained;
  in particular, no new cutvertices emerge in~$G'$ and the separation values of maintained cutvertices do not increase.
\end{proposition}
\begin{proof}
  The only non\-/obvious part is that $B' \df B_1 \cup \hdots \cup B_{k+1}$ is two\-/connected in $G'$,
  which we will prove now.
  Denote $C=(b_1,a,\eli{a}{t},b_1)$ for some $t$.
  There is a cycle of the form $C' = (b_1,\hdots,b_k,\hdots,c,a,\eli{a}{t},u_1)$
  that starts in~$B_1$, runs through $\elix{B}{2}{k+1}$ and finally re\-/enters $B_1$.

  Let $i \in \setn{k+1}$ with $\card{B_i} \geq 3$ and $v \in B_i \setminus V(C')$.
  \textbf{Claim:} in $G'$, there are paths $P=(v,\hdots,x)$ and $Q=(v,\hdots,y)$ with $x,y \in V(C')$ and $x \neq y$
  such that $V(P) \cap V(Q) = \set{v}$ and $V(P) \cap V(C') = \set{x}$ and $V(Q) \cap V(C') = \set{y}$.

  \textbf{Proof of Claim:} Denote $W \df V(C') \cap B_i$, then $\card{W} \geq 2$.
  By \autoref{rem:two-paths}\ref{rem:two-paths:ii} applied with $W$ as defined here
  the claim is clear for $i \geq 2$, since such $B_i$ is two\-/connected in $G'$ (and in~$G$).
  Hence we consider $i=1$.
  The difficulty is that $B_1$ may not be two\-/connected in $G'$ since we removed the edge $a_1u_1$.
  However, none of the paths $P$ or $Q$ guaranteed to exist by \autoref{rem:two-paths}\ref{rem:two-paths:ii} in~$G$
  can use $a_1u_1$ since then that path would have more than one vertex in common with~$W$.
  This concludes the proof of the claim.

  Now let $v, w \in B'$ with $v \neq w$.
  We show that in $G'$, there are two independent $v$-$w$ paths.
  \begin{compactitemize}
  \item If $v,w \in B_i$ for some $i \leq k$,
    then either $B_i$ is two\-/connected and the statement is clear,
    or $\card{B_i}=2$ in which case $v,w \in V(C')$ and the independent paths are given through $C'$.
  \item Let $v \in B_i$ and $w \in B_j$ for $i \neq j$.
    If $v$ or $w$ is located on $C'$, then nothing has to be done for that vertex;
    otherwise we connect it with $C'$ via the paths guaranteed by the claim.
    It is easy to see that this gives two independent $v$-$w$ paths.
  \item Let $v,w \in B_1$.
    In $G$, we find two independent $v$-$w$ paths $P$ and $Q$ in $B_1$.
    At most one of them, say $P$, uses $ab_1$.
    Instead of using that edge we can, starting at $b_1$,
    run along $C'$ until we reach $a$.
    That $b_1$-$a$ path runs outside of $B_1$ (except for $a$ and $b_1$)
    and thus will not interfere with $P$ or $Q$.
    \qedhere
  \end{compactitemize}
\end{proof}

\begin{theorem}
  \label{prop:uniform-vertex-tree}
  An SE for the uniform vertex destroyer is two\-/connected (that is, it has only one block)
  or it does not contain any cycle and thus is a tree.
\end{theorem}
\begin{proof}
  Let an SE graph $G \in \cG_n$ be given and assume it contains more than one block.
  We only need to prove that no block has a cycle, or, equivalently, that each block consists of only two vertices.
  Suppose for contradiction that $B_1$ is a block with $\card{B_1} \geq 3$ and let $B_2$ be another block.
  Let $(a,b,c)$ be a cycle extension with respect to $(B_1,B_2)$ and $G' \df G+(a,b,c)$.
  By \autoref{prop:merge}, it follows that $\rel'(b,a) < \rel(b,a)$ since in $G'$, removal of $b$
  cannot cut $a$ off the one or more vertices in $B_2 \setminus \set{b_1}$ anymore.
  (Recall that a block always contains at least two vertices.)
  All other relevances for $a$ are maintained or also reduced.
  Therefore, we have an improvement in the cost of player $a$, contradicting stability of~$G$.
\end{proof}

\begin{corollary}
  Let $G \in \cG_n$ be an SE for the uniform vertex destroyer.
  Then $G$ is either two\-/connected or a star, hence $\SC(G) = 2(n-1) = \Oh{n}$ or $\SC(G) = 3n-5+\frac{2}{n} = \Oh{n}$.
\end{corollary}
\begin{proof}
  Let $G$ be non\-/two\-/connected.
  Then by \autoref{prop:uniform-vertex-tree}, $G$ is a tree.
  By applying the same argument as in proof of \autoref{lem:linkdiam2},
  we obtain that the diameter of $G$ is at most $2$, \ie $G$ is a star.
  The social cost of star is easily computed to be $\frac{1}{n}\parens[big]{(n-1)(3n-4)+2(n-1)}=3n-5+\frac{2}{n}$.
\end{proof}

A destroyer $\cD$ is called the \term{degree\-/proportional vertex destroyer}
if $\cD_G(u) = \frac{\deg_G(u)}{2m}$ for each $u \in V_n$
and each $G \in \cG_n$.

\begin{proposition}
  The star is an SE for the degree\-/proportional vertex destroyer,
  and its social cost is $\frac{1}{2}(n^2+n)-1 = \Ohm{n^2}$.
\end{proposition}
\begin{proof}
  Let $S \in \cG_n$ be a star.
  First we prove that $S$ is an SE.
  Clearly, by just removing an edge (without creating a new one), no player can improve.
  It is also clear that the only possibility for swapping is from one leaf to another leaf.
  Let $a, c\in V$ be leafs and $b$ the center of the star. Denote $S' \df S + (a,b,c)$.
  Then we have:
  \begin{align*}
    \cost_{S}(a) & = \frac{1}{2(n-1)}\sum_{w \in V} \rel(w,a) \cdot \deg(w) = \frac{1}{2(n-1)}\parens[big]{ (n-1)+(n-2)+(n-1)(n-1) } \\
                 & = \frac{1}{2} \parens[big]{n+1-\frac{1}{n-1}}
  \end{align*}
  After swapping:
  \begin{align*}
    \cost_{S'}(a) = \frac{1}{2(n-1)}\parens[big]{3(n-1)+(n-2)^2+(n-3)}=\frac{1}{2} \parens[big]{n+1-\frac{1}{n-1}}
  \end{align*}
  Hence, $\cost_{S}(a)=C_{S'}(a)$. Therefore, the star is an SE. Its social cost is:
  \begin{align*}
    \SC(S) & = \frac{1}{2(n-1)}\sum_{v\in V}\sum_{w \in V} \rel(w,v) \cdot \deg(w) \\
           & = (n-1)\parens[big]{\frac{1}{2}\parens[big]{n+1-\frac{1}{n-1}}}+\frac{1}{2(n-1)}\cdot(n-1) \, n \\
           & = \frac{1}{2}\parens[big]{(n^2 - 1) - 1}+\frac{n}{2} = \tfrac{1}{2} \parens{n^2+n} - 1
             \tag*{\qedhere}
  \end{align*}
\end{proof}

\begin{corollary}
  The social cost of SE for the degree\-/proportional vertex destroyer can be as high as $\Ohm{n^2}$,
  which is the worst possible order in the destruction model.\qed
\end{corollary}

\section{Extreme Vertex Destruction}

This model is defined similar to the extreme edge destruction model.
Denote $\Vmax$ the set of max\-/sep vertices and $\nmax \df \card{\Vmax}$.
The extreme vertex destroyer picks the vertex to destroy uniformly at random from~$\Vmax$.

We start with a first step toward understanding the worst\-/case order of social cost of an SE in this model,
by giving an SE example with super\-/linear lower bound, namely $\Ohm{n^{3/2}}$.
It is unknown at this time whether there is a matching upper bound.

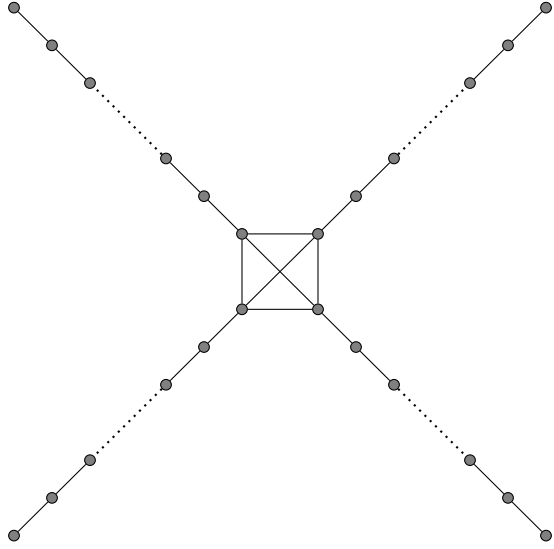
\begin{figure}[t]
  \centering
  \begin{tikzpicture}[on grid]
    \tikzstyle{vertex}=[circle, draw, fill=black!50, inner sep=0pt, minimum width=4pt]
    \foreach \i/\x/\y in {1/1/1,2/-1/1,3/1/-1,4/-1/-1} {
      \foreach \j in {1,2,3,5,6,7} {
        \node[vertex] (v-\i-\j) at (0.5*\x*\j,0.5*\y*\j) {};
        \ifnum \j=1 \relax \else \ifnum \j=5 \relax \else \pgfmathtruncatemacro{\prev}{\j-1} \path (v-\i-\prev) edge (v-\i-\j); \fi \fi
      }
      \path (v-\i-3) edge[dotted,thick] (v-\i-5);
    }
    \path (v-1-1) edge (v-2-1) (v-1-1) edge (v-3-1) (v-1-1) edge (v-4-1);
    \path (v-2-1) edge (v-3-1) (v-2-1) edge (v-4-1);
    \path (v-3-1) edge (v-4-1);
  \end{tikzpicture}
  \caption{Construction from \autoref{prop:lower-extreme} for $t=4$.}
\end{figure}

\begin{theorem}
  \label{prop:lower-extreme}
  Let $t \geq 4$ and $0 \leq k \leq 4t-5 = \The{t}$.
  Let $G=(V,E)$ be the graph consisting of a clique $C$ on $t$ vertices
  and to each vertex of $C$ there is a path of length $k$ attached (so $n \df \card{V} = t (k+1)$).
  Then $G$ is an SE with $\SC(G) = \Ohm{n^{3/2}}$.
\end{theorem}
\begin{proof}
  For each player $v$ at distance $0 \leq i \leq k$ from $C$, we have:
  \begin{equation*}
    \sep(v) = 2\parens[big]{ (n-1) + (k-i) \, (n-1-(k-i)) }
  \end{equation*}
  Since $k-i \leq (n-1)/2$ and since the function $x \mapsto x (n-1-x)$ is strictly increasing on $[0,(n-1)/2]$,
  separation is strictly largest when $i=0$, that is, when $v \in C$.
  It follows $\Vmax = C$ and $\SC(G) = 2\parens{ (n-1) + k (n-1-k) } = \Ohm{tk^2}$.
  If we choose $k$ maximal, then $k=\The{t}$, hence $\SC(G) = \Ohm{t^3} = \The{n^{3/2}}$.

  We prove the SE property.
  Just removing an edge (without building a new one) is clearly not an option for the players on the paths.
  For a player of the clique $C$, also nothing changes when removing an edge
  since due to $t \geq 4$, the set $C$ remains two\-/connected.

  For $u \in C$ denote $V_u$ the $k$ vertices on the path attached to $u$
  (note that $u \not\in V_u$).
  In the following, whenever we consider a swap, by $G'$ we refer to the graph we obtain from $G$ by applying the swap.

  We start with the different swaps available to a player $a \in C$.
  We have $\rel(u,a) = k+1$ for each $u \in \Vmax \setminus \set{a}$, hence
  \begin{equation*}
    \cost(a) = \frac{(n-1)+(t-1)(k+1)}{t} \period
  \end{equation*}

  \begin{enumerate}[label=(\roman*)]
  \item
    Consider a swap $(a,b,c)$ with $b \in C \cap N(a)$ and $c \in V_u$ for $u \in C \setminus \set{a}$,
    that is, player $a$ swaps an edge from the clique to some path but not the one connected to $a$ itself.
    Since $t \geq 4$, the set $C$ remains two\-/connected.
    Separation of $u$ and of some vertices in $V_u$ is reduced.
    All other separations are maintained.
    It follows:
    \begin{equation*}
      \cost(a) - \cost'(a) = \frac{(n-1)+(t-1)(k+1)}{t} - \frac{(n-1)+(t-2)(k+1)}{t-1} = \frac{k-n+2}{t(t-1)} < 0
    \end{equation*}
    Hence the swap is no improvement for player $a$.
  \item
    Consider a swap $(a,b,c)$ with $b \in C \cap N(a)$ and $c \in V_a$,
    that is, player $a$ swaps an edge from the clique to its own path.
    Again, since $t \geq 4$, the set $C$ remains two\-/connected.
    The separation values on some vertices in $V_a$ decrease, but that does not change the set of max\-/sep vertices
    nor their relevance for $a$.
    Hence player $a$'s cost is maintained.
  \item\label{item:lower-extreme:ii}
    Consider a swap $(a,b,c)$ with $b \in V_a$; then in order to keep the graph connected, we have $c \in V_a$.
    That is, player $a$ swaps the first edge on its path to some vertex on that path.
    We only have to exclude that $c$ becomes max\-/sep; if we achieve that, then we know that $a$'s cost is maintained.
    Let $2 \leq l \leq k-1$ be the distance between $a$ and $c$. Then by \eqref{eq:sep-beta}:
    \begin{align*}
      & \sep'(c) < \sep_{\max} \\
      \iff & n^2 - 1 - (l-1)^2 - (k-l)^2 - (n-k)^2 < n^2 - 1 - k^2 - (n-1-k)^2 \\
      \iff & k^2 + (n-1-k)^2 < (l-1)^2 + (k-l)^2 + (n-k)^2 \\
      \iff & (n-1-k)^2 < - 2 l (k+1-l) + 1 + (n-k)^2 \\
      \iff & -2(n-k) < - 2 l (k+1-l) \\
      \iff & n-k > l (k+1-l) \impliedby n-k > \frac{(k+1)^2}{4} \\
      \iff & t(k+1)> \frac{(k+1)^2}{4} + k \iff t > \frac{k+1}{4} + 1 - \frac{1}{k+1} \\
      \impliedby & 4 t - 5 \geq k
    \end{align*}
    The latter is true by the restriction on $k$ in the statement of the theorem.
    In this computation, we used again that the function $x \mapsto x(k+1-x)$ is increasing on $[0,(k+1)/2]$.
  \end{enumerate}

  We continue with the swaps available to a player $a \in V_u$ for some $u \in C$.

  \begin{enumerate}[label=(\roman*),resume]
  \item Consider a swap $(a,b,c)$ with $\dist(u,a) < \dist(u,b)$.
    Then $c \in V_u$ with $\dist(u,b) < \dist(u,c)$, since otherwise the graph would become disconnected.
    From a computation like in \ref{item:lower-extreme:ii}, it follows that this swap cannot make $c$ max\-/sep.
    Hence the cost of player $a$ does not change.
  \item Consider a swap $(a,b,c)$ with $\dist(u,a) > \dist(u,b)$.
    If $c \in V_u$, then again $a$'s cost will not change.
    If $c \in C$, then $c$ will become the only max\-/sep vertex in the new graph,
    clearly increasing $a$'s cost.

    Now let $c \in V_w$ for some $w \neq u$, that is, some vertices migrate from $u$'s path to $w$'s path.
    Separation of $w$ and separation of the vertices $v \in V_w$ with $\dist(w,v) \leq \dist(w,c)$ will increase.
    All other separations are reduced or maintained, so we have $\Vmax' \subseteq V_w \cup \set{w}$.
    In the best case, $k$ vertices migrate to $w$'s path, only $w$ becomes max\-/sep,
    and $\rel'(w,a) = n-2k$.
    In this case:
    \begin{align*}
      \cost(a) - \cost'(a)
      & = \frac{(t-1) (k+1) + (n-k)}{t} - (n-2k) \\
      & = \frac{(t-1) (k+1) + (t(k+1)-k) - t (t(k+1)-2k)}{t} \\
      & = \frac{(k+1) (2 t-t^2-1) + (2t-1) k}{t} \\
      & \leq (k+1) (2-t) + 2 k \leq -2 (k+1) + 2 k < 0
    \end{align*}
    Hence the swap is no improvement for player $a$.\qedhere
  \end{enumerate}
\end{proof}

\begin{remark}
  The graph in \autoref{prop:lower-extreme} is no SE if $k \geq 4t-4 = 4(t-1)$,
  that is, if $k$ is larger than the upper bound in the theorem.
\end{remark}
\begin{proof}
  Let $a \in C$ and $b$ her neighbor in $V_a$.
  Let $c \in C$ be at distance $l$ from $a$, to be specified later.
  Denote $G' \df G + (a,b,c)$.
  If $c$ becomes max\-/sep in $G'$, then player $a$'s cost will decrease
  since $\rel'(c,a) = k$, whereas the relevance of a max\-/sep vertex in $G$ for $a$ is $k+1$ or $n-1$.
  By the computation in \autoref{prop:lower-extreme}\ref{item:lower-extreme:ii},
  we see that $c$ becomes max\-/sep if $n-k \leq l (k+1-l)$. We have:
  \begin{equation*}
    n-k \leq l (k+1-l) \iff
    t \leq \frac{l (k+1-l)}{k + 1} + 1 - \frac{1}{k+1}
    \impliedby t \leq \frac{l (k+1-l)}{k + 1} + \frac{4}{5}
  \end{equation*}
  For odd $k \geq 4(t-1)$ and $l \df \frac{k+1}{2}$, we have:
  \begin{align*}
    t \leq \frac{l (k+1-l)}{k + 1} + \frac{4}{5}
    & \iff t \leq \frac{(k+1)^2/4}{k + 1} + \frac{4}{5}
      \iff t \leq \frac{k+1}{4} + \frac{4}{5} \\
    & \impliedby t \leq \frac{4t-3}{4} + \frac{4}{5}
      \iff 0 \leq \frac{-3}{4} + \frac{4}{5} \iff 0 \leq \frac{1}{20}
  \end{align*}
  For even $k \geq 4(t-1)$ and $l \df \frac{k}{2}$, we have:
  \begin{align*}
    t \leq \frac{l (k+1-l)}{k + 1} + \frac{4}{5}
    & \iff t \leq \frac{\frac{k^2}{4} + \frac{k}{2}}{k + 1} + \frac{4}{5}
      \iff t \leq \frac{k^2+k}{4(k + 1)} + \frac{k}{4(k+1)} + \frac{4}{5} \\
    & \impliedby t \leq \frac{k}{4} + \frac{1}{4} \frac{4}{5} + \frac{4}{5}
      \iff t \leq \frac{k}{4} + 1
      \impliedby t \leq t-1 + 1
      \tag*{\qedhere}
  \end{align*}
\end{proof}

In the remainder of this work, we provide more insight into the structure of tree SE graphs for the extreme vertex destroyer.

\begin{theorem}
  There is no SE tree with $n_{\max}=1$, provided that $n \geq 8$.
\end{theorem}
\begin{proof}
  Let $T=(V,E)$ be an SE tree with $n_{\max}=1$ and $u\in V_{\max}$. It is clear that $u$ is a cutvertex,
  otherwise $G$ is two\-/connected and thus $n_{\max}=n$.
  Denote $\eli{K}{k}$, with $k \geq 2$, the components of $T-u$
  ordered by non\-/decreasing sizes $|K_1|\leq|K_2|\leq   \hdots\leq |K_k|$.
  For convenience, denote $n_i \df \card{K_i}$ for each~$i$.

  Let $v\in K_1$ and $uv\in E$ and $w \in K_2$ with $\deg(w)=1$.
  We consider $T' \df T + (v,u,w)$,
  that is, we detach $K_1$ from $u$ and re-attach it to a leaf of $K_2$.

  Then $\rel'(w,v) = \rel(u,v)$, and so we have an improvement for $v$
  whenever $\Vmax'$ contains at least one vertex distinct from $w$.
  The latter is the case if $\sep'(w) \leq \sep'(u)$.
  Using \eqref{eq:sep-beta}, we compute:
  \begin{align*}
    \sep'(w) \leq \sep'(u)
    & \iff n^2 - 1 - n_1^2 - (n-n_1-1)^2 \leq n^2 - 1 - (n_1+n_2)^2 - \sum_{i=3}^k n_i^2 \\
    & \iff n_1^2 + (n-n_1-1)^2 \geq (n_1+n_2)^2 + \sum_{i=3}^k n_i^2 \\
    & \iff (n-n_1-1)^2 \geq 2n_1n_2 + \sum_{i=2}^k n_i^2 \\
    & \iff \parens{\sum_{i=2}^k n_i}^2 \geq 2n_1n_2 + \sum_{i=2}^k n_i^2 \\
    & \iff \sum_{2 \leq i < j \leq k}n_in_j \geq n_1n_2 \tag{$*$}\label{eq:tree-nmax-1-condition}
  \end{align*}
  Now, \eqref{eq:tree-nmax-1-condition} is true if $k \geq 3$, since $n_1 \leq n_3$.

  Hence we may assume that $k = 2$.
  Let $v$ be the only vertex in $N(u) \cap K_2$.
  If $\deg(v) \geq 3$, then:
  \begin{align*}
    \sep(v) & \geq n^2 - 1 - 1 - (n_2 - 2)^2 - (n_1 + 1)^2 \\
            & = n^2 - 1 - 1 - n_2^2  + 4n_2 - 4 - n_1^2 - 2n_1 - 1 \\
            & = \sep(u) - 1 + 4n_2 - 4 - 2n_1 - 1 \\
            & = \sep(u) + 2 (2n_2 - n_1) - 6 \\
            & \geq \sep(u) + 2 n_2 - 6
              \geq \sep(u)
  \end{align*}
  The last step is true since $n_2 \geq (n-1) / 2 > 3$.
  We conclude that $\deg(v) \leq 2$.
  Since $n \geq 8$, there is $w \in N(v) \setminus \set{u}$ with $\deg(w) \geq 2$.
  We consider $T' := T + (u,v,w)$. We have:
  \begin{align*}
    \sep'(u) < \sep'(w) & \impliedby n_1^2 + n_2^2 > 1 + (n_2 - 2)^2 + (n_1 + 1)^2 \\
                        & \iff 0 > 1 - 4n_2 + 4 + 2n_1 + 1 \iff 0 > 2 (n_1-2n_2) + 6 \\
                        & \iff 0 > n_1-2n_2 + 3 \impliedby n_2 > 3
  \end{align*}
  Since the last statement is true, we know that $\Vmax' = \set{w}$, hence $\cost'(u) = n_2 < n-1 = \cost(u)$.
\end{proof}

\end{document}